\documentclass[11pt]{article}
\usepackage{latexsym}
\usepackage{amsmath}
\usepackage{amssymb}
\usepackage{amsthm}
\usepackage{epsfig}
\usepackage{fullpage}
\usepackage{longtable}
\usepackage{hyperref}
\newtheorem{theorem}{Theorem}[section]

\newtheorem{lem}[theorem]{Lemma}
\newtheorem{obs}[theorem]{Observation}

\newtheorem{definition}[theorem]{Definition}
\newtheorem{claim}[theorem]{Claim}
\newtheorem{fact}[theorem]{Fact}
\newtheorem{remark}[theorem]{Remark}
\newcommand{\lex}{lexicographic }
\newcommand{\lexq}{{<}_{\mathsf{lex}}}
\newcommand{\p}{\mathsf {Prefix}}
\newcommand{\s}{\mathsf {Suffix}}
\newcommand{\st}{\mathsf {Substring}}
\newcommand{\np}{\mathsf{P}}
\newcommand{\ns}{\mathsf{S}}

\newcommand{\Z}{\mathbb{Z}}
\newcommand{\F}{\mathbb{F}}

\newcommand{\Rot}{\mathsf{Rot}}
\newcommand{\fp}{\mathsf{FP}}
\newcommand{\Orbit}{\mathsf{Orbit}}
\newcommand{\Classes}{\mathcal{C}}

\newcommand{\poly}{\mathsf{poly}}
\newcommand{\done}{{\times}}
\newcommand{\doing}{{\checkmark}}
\newcommand{\wt}{{\mathsf{wt}}}

\newcommand{\defeq}{\stackrel{\text{def}}{=}}

\newcommand{\bin}{\mathsf{Bin}}
\begin{document}

\title{Efficient Indexing of Necklaces and Irreducible Polynomials over Finite Fields}
\author{Swastik Kopparty\thanks{Department of Computer Science and Department of Mathematics, Rutgers University. Research supported in part by a Sloan Fellowship and NSF grant CCF-1253886.
Email: \texttt{swastik.kopparty@gmail.com}.}\and Mrinal Kumar\thanks{Department of Computer Science, Rutgers University. Research supported in part by NSF grant CCF-1253886.
Email: \texttt{mrinal.kumar@rutgers.edu}.}\and Michael Saks\thanks{Department of Mathematics, Rutgers University. Research supported in part by NSF grants CCF-0832787  and CCF-1218711. Email: \texttt{msaks30@gmail.com}}}

\maketitle

\begin{abstract}

We study the problem of {\em indexing} irreducible polynomials over finite fields,
and give the first efficient algorithm for this problem.
Specifically, we show the existence of $\poly(n, \log q)$-size
circuits that compute a bijection between $\{1, \ldots, |S|\}$ and the set
$S$ of all irreducible, monic, univariate polynomials of degree $n$ over a finite field $\F_q$.
This has applications in pseudorandomness, and answers
an open question of Alon, Goldreich, H{\aa}stad and Peralta~\cite{AGHP}.

Our approach uses a connection between irreducible polynomials and necklaces ( equivalence classes of strings under cyclic rotation).
Along the way, we give the first
efficient algorithm for indexing necklaces of a given length
over a given alphabet, which may be of independent interest.

\end{abstract}

\section{Introduction}

For a finite field $\F_q$ and an integer $n$,
let $S$ be the set of all irreducible polynomials in $1$ variable over $\F_q$
of degree exactly $n$. There is a well known formula for $|S|$
(which is approximately $\frac{q^n}{n}$).
We consider the problem of giving an efficiently computable {\em indexing} of  irreduducible polynomials
i.e., finding a bijection $f: \{1, \ldots, |S| \} \to S$ such that
$f(i)$  is computable in time $\poly(\log |S|) = \poly(n \log q)$.
Our main result is that indexing of irreducible polynomials can be done efficiently
given $O(n \log q)$ advice. This answers a problem posed by Alon, Goldreich, H{\aa}stad and Peralta~\cite{AGHP},
and is the polynomial analogue of the the well-known problem of ``giving a formula for the $n$-bit primes".
Note that today it is not even known (in general) how to produce a single irreducible polynomial of degree
$n$ in time $\poly(n \log q)$ without the aid of either advice or randomness.

The main technical result we show en route is an efficient indexing algorithm for {\em necklaces}.
Necklaces are equivalance classes of strings modulo cyclic rotation.
We give an $\poly(n \log |\Sigma|)$-time
computable bijection $g: \{1, 2, \ldots, |\mathcal N|\} \to \mathcal N$,
where $\mathcal N$ is the set of necklaces of length $n$ over the alphabet $\Sigma$.

\subsection{The indexing problem}

We define an {\em indexing} of a finite set $S$ to be a bijection from
the set $\{1,\ldots,|S|\}$ to $S$. 
Let us formalize indexing as a computational problem.   
Suppose that $L$ is an arbitrary
language over alphabet $\Sigma$ and let $L^n$ be the
set of strings of $L$ of length $n$.      We want to ``construct'' an indexing
function $A^n$ for each of the sets $L^n$.  Formally, this means
giving an algorithm $A$ which takes as input a size parameter $n$
and an index $j$ and outputs $A^n(j)$, so that the following
properties hold for each $n$:

\begin{itemize}
\item $A^n$ maps the set $\{1,\ldots,|L^n|\}$ bijectively to $L^n$.
\item If $j > |L^n|$ then $A^n(j)$ returns {\bf too large}.
\end{itemize}  

An indexing algorithm is considered to be efficient if its running time is $\poly(n)$.

A closely related problem is {\em reverse-indexing}. A
reverse-indexing of $L$ a bijection from $L^n$ to $\{1, \ldots, |L^n|\}$, and we say it is efficient if it can be computed in time $\poly(n)$.

We can use the above formalism for languages to
formulate the indexing and reverse-indexing problems for any combinatorial structure, such as permutations, graphs, partitions, etc. by using standard efficient encodings of such structures by strings.

\subsection{Indexing, enumeration, counting and ranking}

Indexing is closely related to the well-studied {\em counting}, {\em enumeration} and {\em ranking} problems for $L$.
The counting problem is to give an algorithm that, on input $n$ outputs
the size of $L^n$. The enumeration problem is to give an algorithm that, on input $n$, outputs a list containing all elements of $L^n$. A counting or enumeration algorithm is said to be efficient if it runs in time $\poly(n)$ or $|L^n| \cdot \poly(n)$ respectively.

Other important algorithmic problems associated with combinatorial objects
include the {\em ranking} and {\em unranking} problems.
For the ranking problem, one is given an ordering of $L^n$ (such as the
lexicographic order) and the goal is to compute the rank (under this order) of a
given element of $L^n$. For the unranking problem, one has to compute the inverse of this ranking
map.
It is easy to see that unranking algorithms for any ordering are automatically indexing algorithms,
and ranking algorithms for any ordering are automatically reverse-indexing
algorithms\footnote{We use the terms indexing and reverse-indexing instead of the terms unranking and ranking
to make an important distinction: in indexing and reverse-indexing the actual bijection
between $\{1, \ldots, |S| \}$ and $S$ is of no importance whatsoever, but in ranking and unranking
the actual bijection is part of the problem. We feel this difference is worth highlighting, and hence
we introduced the new terms indexing and reverse-indexing for this purpose. Note that some important prior
work on ranking/unranking distinguishes between these notions~\cite{MR-perm-rank}.}.

There is well developed complexity theory for counting problems, starting with the fundamental work of Valiant~\cite{Valiant}. For combinatorial structures, counting problems are  (of course) at the heart of combinatorics, and many basic identities in combinatorics  (such as recurrence relations that express the number
of structures of a particular size in terms of the number of such structures of smaller sizes) can also be viewed as giving efficient counting algorithms for
these structures.

The enumeration and ranking problems for combinatorial structures has also received a 
large amount of attention.
See the books~\cite{NW78, stinson1999combinatorial, ruskey2003combinatorial, arndt2011matters}
for an overview of some of the work on this topic.


Counting and enumeration can be easily reduced to indexing:  Given an indexing algorithm
$A$ we can compute $|L^n|$ by calling $A^n(j)$ on increasing powers of
2 until we get the answer `{\bf too large}' and then do binary search to
determine the largest $j$ for which $A^n(j)$ is not too large. Enumeration
can be done by just running the indexing algorithm on the integers
$1, 2, \ldots$ until we get the answer {\bf too large}.

Conversely, in many cases, such as for subsets, permutations, set partitions, integer partitions, trees, spanning trees, (and many many more) the known
counting algorithms can be modified to give
efficient indexing (and hence enumeration) algorithms.  This happens, for example, when the counting problem is solved by a recurrence relation that is proved via a bijective proof.

However, it seems that not all combinatorial counting arguments lead to efficient indexing algorithms.
A prime example of this situation is when we have a finite group acting on a finite set, and the
set we want to count is the set of orbits of the action.
The associated counting problem can be solved
using the Burnside counting lemma, and there seems to be no general way to
use this to get an efficient indexing algorithm.

This leads us to one of the indexing problems studied here: Fix an alphabet $\Sigma$
and consider two strings $x$ and $y$ in  $\Sigma^n$ to be equivalent
if one is a rotation of the other, i.e. we can find strings $x^1,x^2$
such that $x=x^1 x^2$ and $y=x^2x^1$ (here $uv$ denotes
the concatenation of the strings $u$ and $v$). 
The equivalence classes of strings are precisely the orbits
under the natural action of the cyclic group $\Z_n$ on $\Sigma^n$.
These equivalence classes are often called {\em necklaces} because
if we view the symbols of a string as arranged in a circle, then
equivalent strings give rise to the same arrangement. We are interested
in the problem of efficiently indexing necklaces. We apply the indexing algorithm for necklaces to the problem of indexing irreducible polynomials over a finite field.

\subsection{Main results}

Our main result is an efficient algorithm for indexing
irreducible polynomials.

\begin{theorem}
Let $q$ be a prime power, and let $n\geq 1$ be an integer.
Let $I_{q,n}$ be the set of monic irreducible polynomials of 
degree $n$ over $\F_q$.

There is an indexing algorithm for $I_{q,n}$,
which takes $O(n \log q)$ bits of advice
and runs in $\poly(n, \log q)$ time.
\end{theorem}

We remark that it is not known today how to deterministically
produce (without the aid of advice or randomness) even a single irreducible polynomial of degree $n$
in time $\poly(n \log q)$ for all choices of $n$ and $q$. Our result
shows that once we take a little bit of advice, we can produce not
just one, but all irreducible polynomials.
For constant $q$, where it is known how to deterministically construct a single
irreducible polynomial in $\poly(n)$ time without advice~\cite{Shoupirred},
our indexing algorithm can be made to run with just
$\poly(\log n)$ bits of advice.

Using a known correspondence~\cite{Golumb}
between necklaces and irreducible polynomials
over finite fields,
indexing irreducible polynomials reduces to the problem of indexing necklaces.
Our main technical result (of independent interest)
is an efficient algorithm for this latter problem.


\begin{theorem}~\label{thm:indexing}
There is an algorithm for indexing necklaces of length $n$ over the alphabet $\{1,\ldots,q\}$,
which runs in time $\poly(n\log q)$.
\end{theorem}

Our methods also give an efficient reverse-indexing algorithm for necklaces (but unfortunately
this does not lead to an efficient reverse-indexing algorithm for irreducible polynomials; this has to do with the the open problem of efficiently computing the discrete logarithm).

\begin{theorem}~\label{thm:reverse_indexing}
There is an algorithm for reverse-indexing necklaces of length $n$ over the alphabet $\{1,\ldots,q\}$,
which runs in time $\poly(n\log q)$.
\end{theorem}

The indexing algorithm for irreducible polynomials can be used 
to make a classical $\epsilon$-biased set construction
from~\cite{AGHP} based on linear-feedback shift register sequences
constructible with logarithmic advice (to put it at par with the
other constructions in that paper). It can also be used
to make the explicit subspace designs of~\cite{GK-subs-design}
very explicit (with small advice).

Agrawal and Biswas~\cite{AB03} gave a construction of a family of nearly-coprime 
polynomials, and used this to give randomness-efficient black-box
polynomial identity tests. The ability to efficiently
index irreducible polynomials enables one to do this even more
randomness efficiently (using a small amount of advice). 

Similarly, the string fingerprinting algorithm by Rabin~\cite{Rabin81}, which is based on choosing a random
irreducible polynomial can be made more randomness efficient by choosing the random irreducible polynomial
via first choosing a random index and then indexing the corresponding irreducible polynomial using our indexing algorithm.
This application also requires a small amount of advice.  

As another application of the indexing algorithm for necklaces,
we give a $\poly(n)$ time algorithm for computing any given entry of the 
$k \times 2^n$ generator matrix matrix or the $(2^n-k) \times 2^n$
parity check matrix of BCH codes for all values of the
designed distance (this is the standard
notion of strong explicitness for error-correcting codes).
Earlier, it was only known how to compute this entry explicitly
for very small values of the designed distance (which is usually
the setting where BCH codes are used).

\subsection{Related Work}

There is an  extensive literature on enumeration
algorithms for combinatorial objects (see the books
\cite{ruskey2003combinatorial,knuth-4,stinson1999combinatorial,
NW78,arndt2011matters}). Some of these
references discuss necklaces in depth, and some also
discuss the ranking/unranking problems for various combinatorial
objects.

The lexicographically smallest element of a rotation class
is called a {\em Lyndon word}, and much is known about them.
Algorithmically, the problem of enumerating/indexing necklaces
is essentially equivalent to the problem of enumerating/indexing 
Lyndon words. Following a long line of work
\cite{neckgen1,neckgen2,neckgen3,neckgen4,neckgen5,RS-necklace,CRSSM00}, we now know linear time enumeration algorithms for Lyndon words/necklaces.

In~\cite{martinez2004efficient} and~\cite{ruskey2003combinatorial}, it was noted that
the problem of efficient ranking/unranking of the lexicographic order
on Lyndon words is an open problem.
Our indexing algorithms in fact give a solution to this problem too: 
we get an efficient ranking/unranking algorithm for the
lexicographic order on Lyndon words.

Recent work of Andoni, Goldberger, McGregor and Porat~\cite{DBLP:conf/stoc/AndoniGMP13} studied a problem that may be viewed as an approximate version of reverse indexing of necklaces. They gave a randomized algorithm for producing short fingerprints of strings, such that the fingerprints of rotations of a string are determined by the fingerprint of the string itself. This fingerprinting itself was useful for detecting proximity of strings under misalignment.

\paragraph{Recent independent work :}
A preliminary version of this paper appeared as~\cite{KKS14}.  At about
the same time, similar results were published by Kociumaka, Radoszewski and  Rytter \cite{KRR14}.  The work in these two papers was done independently.  The papers both have polynomial
time algorithms for indexing necklaces; the authors in~\cite{KRR14} exercised
more care in designing the algorithm to obtain a better polynomial running time.  Their approach to alphabets of size more than 2 is  cleaner than ours. On the other hand, we put the results in a broader context and have some additional applications (indexing irreducible polynomials
and explicit constructions).

\subsection{Organization of the paper}
 The rest of the paper is organized as follows. We give the algorithm to index
necklaces in Section~\ref{sec:countequiv}. In Section~\ref{sec:irrind}, we use our indexing algorithm for necklaces to give an indexing algorithm for irreducible polynomials over finite fields. In Section~\ref{sec:BCH}, we give an application to the explicit construction of generator and parity check matrices of BCH codes. We conclude with some open problems in Section~\ref{sec:openprobs}.
In Appendix~\ref{sec:magic}, we give an alternate algorithm for indexing binary necklaces of prime length. In Appendix~\ref{sec:complexity}, we give some prelimary observations about the complexity theory of indexing in general.

\section{Indexing necklaces}\label{sec:countequiv}

\subsection{Strategy for the algorithm}~\label{sec:strategy}

We first consider a very basic indexing algorithm which will inspire our algorithms.
Given a directed acyclic graph $D$ on vertex set $V$ and  distinguished subsets $S$ and $T$ of nodes,
there is a straightforward indexing algorithm for the set of of paths that start in $S$ and end in $T$:  
Fix an arbitrary
ordering on the nodes, and consider the induced lexicographic ordering on paths 
(i.e. path $P_1P_2\ldots$ is less than path $Q_1Q_2\ldots$ if  $P_i<Q_i$ where $i$ is the least integer
such that $P_i \neq Q_i$).  Our indexing function
will map the index $j$ to the $j$th path from $S$ to $T$ in lexicographic order.  There is a simple
dynamic program which computes for each node $v$, the number $N(v)$ of paths from $v$ to a vertex in $T$.
Let $v_1,\ldots,v_r$ be the nodes of $S$ listed in order.
Given  the input index $j$, we find the first source $v_i$ such that the number of paths to $T$
starting at nodes $v_1,\ldots,v_i$ is at least $j$; if there is no such source then
the index $j$ is larger than the number of paths being indexed.  Otherwise,   $v_i$ is the first node
of the desired path, and we can proceed inductively by replacing the set $S$ by the set of
children of $v_i$.

This approach can be adapted to the following situation.  Suppose the
set $S$ we want to index is a set of strings of fixed length $n$ over alphabet $\Sigma$.
A {\em read-once branching program} of length $n$ over alphabet $\Sigma$ is an acyclic
directed graph with vertex layers numbered from 0 to $n$, where (1)  layer 0 has
a single start node, (2) there is a designated subset of accepting nodes at level $n$,
and (3) every non-sink node has one outgoing arc corresponding
to each alphabet symbol, and these arcs connect the node
to  nodes at the next level.   For nodes $v$ and $w$ and alphabet
symbol $\sigma$ we write $v\rightarrow_{\sigma} w$ to mean that there is an arc
from $v$ to $w$ labelled by $\sigma$.

 Such a branching program
takes words from $\Sigma^n$ and, starting from the start node,
follows the path corresponding to the word to either the accept or reject node.
Given a read-once branching program for $S$, there is a 1-1 correspondence
between strings in $S$ and paths from the start node to an accepting node.
We can use the indexing algorithm for paths given above to index $S$.

This suggests the following approach to indexing necklaces.  For each equivalence class of strings (necklace)
identify a canonical representative string of the class (such as the lexicographically smallest
representative).   Then build a branching program $B$ which, given string $y$, determines
whether $y$ is a canonical representative of its class.
By the preceding paragraph, this would be enough to index all of the canonical
representatives, which is equivalent to indexing equivalence classes.

In fact, we are able to implement this approach provided that $q=2$ 
and $n$ is prime (See appendix \ref{sec:magic}).
However, we have not been able to make it work in general.  For this we need another approach,
which still uses branching programs, but in a more involved way.

First some notation.
For a given string $y$, we write the string obtained from $y$ after cyclically rotating it rightwards  by $i$ positions as $\Rot^i(y)$. We define 
$\Orbit(y)$ to be the set containing $y$ and all its distinct rotations. $\Orbit(y)$ will also be referred to as the {\it equivalence class} of $y$. A string $y$ is said to be periodic with period $p$ if it can be written as ${y_1}^{q}$ for some ${y_1}\in \Sigma^{p}$ and $q = \frac{n}{p}$. A string is said to have fundamental period $p$ if it is periodic with period $p$ and not periodic with any period smaller than $p$. We will denote the fundamental period of a string $y$ by $\fp(y)$. Note that for any string $y$,  $|\Orbit(y)| = \fp(y)$.

If $E$ is an orbit and $x$ is a string, we say that $E<x$ if
$E$ contains at least one string $y$ that is lexicographically less than $x$.
(Notice that under our definition, if $x$ and $y$ are strings then
we might have both that the orbit of $x$ is less than $y$ and the orbit of $y$ is less than $x$).

Let $t$ be the total number of orbits.
Let $\Classes_x$ be the set of orbits that are less than $x$.
Our main goal will be to design an efficient algorithm which, given string $x$, returns
$|\Classes_x|$.   We now show that if we can do this then we can solve
both the indexing and reverse indexing problems.

For the indexing problem, we want a 1-1 function $\psi$ that maps 
$j \in\{1,\ldots,t\}$ to a string so that all of the image strings are in different orbits.
The map $\psi$ will be easily computabile given a subroutine for $|\Classes_x|$.

Define the {\em minimal representative} of an orbit to be the lexicographically least string in the orbit.
Let $y^1< \cdots < y^t$ denote the minimal representatives in lex order.   Our map $\psi$ will map $j$ to $y^j$.
This clearly maps each index to a representative of a different orbit.

It suffices to show how to compute $\psi(j)$.  
Note that $|\Classes_x|$ is equal to the number of $y^i$ that precede $x$, and is thus a nondecreasing function of $x$.
Therefore, $\psi(j)=y^j$ is equal to the
lexicographically largest string with $|\Classes_x| <j$.    Furthermore, since $|\Classes_x|$
is a nondecreasing function of $x$,  we can find $\psi(j)$ by doing binary search on the set of strings according
to the value of $|\Classes_x|$.

Simiarly, we can solve the reverse indexing problem: given a string $x$ we can find the index of the orbit to which it belongs
by first finding the lexicographically minimal representative $y^i$  of its orbit  and then computing $|\Classes_{y^i}|+1$.

\begin{lem}~\label{lem:reduction}
To efficiently index and reverse index necklaces of length $n$ over an alphabet $\Sigma$, it suffices to have an efficient algorithm that takes as input a string $x \in \Sigma^n$ and outputs $|\Classes_x|$. 
\end{lem}

The next section gives our algorithm to determine $|\Classes_x|$ fpr any input string $x$.

\subsection{Computing $|\Classes_x|$}

Let us define:
\begin{itemize}
\item $G_{x,p} = \bigcup_{ E \in \Classes_x: |E| =  p} E$.
\item $G_{x, \leq p} = \bigcup_{ E \in \Classes_x: |E| \mbox{ divides }  p} E$.
\end{itemize}

In Section ~\ref{sec:reduc} we reduce the problem of computating of
$|\Classes_x|$ to the problem of computing $|G_{x,\leq p}|$ for various $p$. 
The main component of the indexing algorithm is a subroutine
that computes $|G_{x, \leq p}|$ given a string $x$
and an integer $p$.
This subroutine works by building a branching program with
$n^{O(1)}$ nodes, which when given a string $y$
accepts if and only if (1) the orbit of $y$ has size dividing $p$ and
(2) $\Orbit(y) < x$.
The quantity we want to compute, $|G_{x, \leq p}|$, is therefore
simply the number of $y$ accepted by this branching program  (which, as noted above can be
computed in polynomial time via a simple dynamic program).

\subsubsection{Notation and Preliminaries}~\label{sec:notations}






\noindent
{\bf Preliminaries:}
We state some basic facts about periodic strings without proof. 

\begin{fact}\label{obs:periodic1}
Let $y$ be a string of length $n$ and let $p$ be positive integer dividing $n$. Then, $|\Orbit(y)| = p$ if and only if  $y$  has fundamental period  $p$. In particular, $y$ can be written as ${y_1}^{\frac{n}{p}}$ for an aperiodic string $y_1 \in \Sigma^p$. 
\end{fact}

\begin{fact}~\label{obs:periodic2}
The fundamental period of a string is a divisor of any period of the string.
\end{fact}

In particular, the fundamental period of a string is unique. We denote the fundamental period of $y$ by $\fp(y)$.




\subsubsection{Reduction to computing $|G_{x, \leq p}|$}
\label{sec:reduc}

We begin with some simple transformations that reduce the computation of
$|\Classes_x|$ to the computation of $|G_{x, \leq p}|$ (for various $p$).

\begin{lem}\label{lem: 1}
For all $x \in \Sigma^n$, 
$$ |\Classes_x|  = \sum_{y \in G_{x, \leq n}} \frac{1}{|\Orbit(y)|} = \sum_{y \in G_{x, \leq n}} \frac{1}{\fp(y)}.$$
\end{lem}
\begin{proof}
For $y \in G_{x, \leq n}$,  $\Rot^i(y) \in G_{x, \leq n}$ for every positive integer $i$.
Note that there are exactly $|\Orbit(y)|$ distinct strings of the form $\Rot^i(y)$.
Thus for any orbit $E \in \Classes_x$, we have $\sum_{y \in E} \frac{1}{|\Orbit(y)|} = 1$. 
Therefore:
$$ \sum_{y \in G_{x, \leq n}} \frac{1}{|\Orbit(y)|} = \sum_{E \in \Classes_x} \sum_{y \in E} \frac{1}{|\Orbit(y)|} = \sum_{E \in \Classes_x} 1 = |\Classes_x|.$$
\end{proof}

The sum on the right hand side can be split on the basis of the period of $y$. From Lemma~\ref{lem: 1} and Fact~\ref{obs:periodic1}, we have the following lemma.
\begin{lem}\label{lem: 2}
For all $x \in \Sigma^n$, 
$$ |\Classes_x| = \sum_{i\mid n} \frac{|G_{x,i}|}{i}$$
\end{lem}

So, to count $|\Classes_x|$ efficiently, it suffices to compute $|G_{x,i}|$ efficiently for each $i|n$. Now,  from the definitions, we have the following lemma. 
\begin{lem}~\label{lem: 3}
For all $x \in \Sigma^n$,
$$|G_{x,\leq p}| = \sum_{i|p}|G_{x, i}|$$
\end{lem}

From the M\"obius Inversion Formula (see Chapter 3 in~\cite{Stanley1} for more details), we have the following equality.

\begin{lem}~\label{lem:4}
$$|G_{x,p}| = \sum_{i|p}\mu\left(\frac{p}{i}\right)|G_{x,\leq i}|$$
\end{lem}

Lemma~\ref{lem:4} implies that it suffices to compute $|G_{x,\leq p}|$ efficiently for every divisor $p$ of $n$. In the next few sections, we will focus on this sub-problem and design  an efficient algorithm for this problem. We will first describe the algorithm when the alphabet is binary, and then generalize to larger alphabets. 

\subsubsection{Computing $|G_{x, \leq n}|$ efficiently for the binary alphabet}~\label{sec:binalphabet}

In this section, we will design an efficient algorithm that given a string $x \in \{0,1\}^n$ computes $|G_{x, \leq n}|$.  
On input $x$ the algorithm will construct a branching program with the property that $|G_{x,\leq n}|$ is the
number of accepting paths in the branching program.  This number of accepting paths can be computed
by a simple dynamic program as described at the beginning of Section ~\ref{sec:strategy}.

\begin{lem}~\label{lem:countautomaton}
Given as input a branching program $B$ of length $n$ over alphabet $\Sigma$, we can compute the size of the
set of accepted strings in time $\poly(|B|,\log n)$.
\end{lem}
\begin{proof}
The number accepted strings is the number of paths from the start node to the accept node, and all such paths have length exactly $n$.
Thus the number of accepted strings is the 
$i,j$ entry in the $n^{\rm{th}}$ power of the adjacency matrix of the graph.
and can thus be computed in time polynomial in the size of the graph and $\log n$ (by repeated squaring). 
\end{proof}

We now describe how to construct, for each fixed string $x \in \{0,1\}^n$, a branching program $B_x$ of size polynomial in $n$ such that
the strings accepted by $B_x$ are exactly those in  $G_{x, \leq n}$.
Lemma~\ref{lem:countautomaton} then implies that we can compute $|G_{x,\leq n}|$ in time polynomial in $n$.

For strings $x,y$, when is $y \lexq x$?  This happens and only if there exists an $i\in \{1,2,\ldots, n-1\}$ such that $y_j = x_j$ for every $j\leq i$ and $x_{i+1} > y_{i+1}$.  In the case of binary strings of length $n$, we must have $x_{i+1} = 1$ and $y_{i+1} = 0$.

\begin{definition}
The set of witnesses for $x$, denoted $L_x$, is defined by:
$$L_x = \{s0 : s1\text{ is a prefix of } x \}  $$
\end{definition}
We can summarize the discussion from the paragraph above as follows: 
\begin{obs}~\label{obs:lexineq1}
For $x, y \in \{0,1\}^n$,
we have $y \lexq x$ if and only if some prefix of $y$ lies in $L_x$.  
\end{obs}

 
We will now generalize this observation to  strings under rotation.
For strings $x, y$, when is $\Orbit(y) < x$?  Recall that $\Orbit(y) < x$ if for some $y' \in \Orbit(y)$, we have $y' \lexq x$.
From Observation~\ref{obs:lexineq1}, we know that this happens if and only if some $y' \in \Orbit(y)$ has some prefix $w$ in $L_x$.
Rotating back to $y$, two situations can arise. Either $y$ contains $w$ as a contiguous substring, or $w$ appears as a ``split substring"
wrapped around the end of $y$. In the latter case, $y$ has a prefix $w_1$ and a suffix $w_2$ such that $w_2w_1 = w \in L_x$. 

Recall that $G_{x, \leq n}$ is the set of $y$ with $\Orbit(y) < x$.
Thus,  $y \in G_{x, \leq n}$ if and only if  it  has a contiguous substring as a witness, or it has a witness that is wrapped around its end.
Let us separate these two cases out. 
 
\begin{definition}~\label{def:partition}
For a string $x\in \{0,1\}^n$, 
$$G_{x, \leq n}^c = \{y\in \{0,1\}^n: y \text{ contains a string in } L_x \text{ as a contiguous substring }\}$$
$$G_{x, \leq n}^w = \{y \in \{0,1\}^n: y \text{ has a prefix } w_1\text{ and suffix } w_2 \text{ such that } w_2w_1 \in L_x\}$$

\end{definition} 

From the discussion in the paragraph above, we have the following observation:

\begin{obs}~\label{obs:separate}
$$G_{x, \leq n} = G_{x, \leq n}^c \cup G_{x, \leq n}^w$$
\end{obs}

The branching program $B_x$ will be obtained by combining two branching programs $B_x^c$ and $B_x^w$,
where the first accepts the strings in $G_{x,\leq n}^c$ and the second accepts the strings in $G_{x,\leq n}^w$.
Each layer $j$ of the branching program $B_x$ is the product
of layer $j$ of $B_x^c$ and layer $j$ of $B_x^w$ and we have arcs $(v,v')\rightarrow_{\sigma} (w,w')$  when $v \rightarrow_{\sigma} v'$
and $w \rightarrow_{\sigma} w'$.  The accepting nodes at level $n+1$ are nodes $(v,v')$ where $v$ is an accepting node
of $B_x^c$ or $v'$ is an accepting node of $B_x^w$.  The resulting branching program clearly accepts the set of strings
accepted by $B_x^c$ or $B_x^w$.

Note that the branching programs $B_x$ produced by the algorithm are never actually ``run'', but are given as input
to the algorithm of Lemma~\ref{lem:countautomaton} in order to determine $|G_{x,\leq n}|$.

For a set of strings $W$, we will use $\p(W)$ to denote the set of all prefixes of all strings in $W$ (including the empty string $\epsilon$). Similarly, $\s(W)$ denotes the set of all suffixes of of all strings in $W$ (including the empty string $\epsilon$). 
Similarly, we will use $\st(W)$ for set of all contiguous substrings of strings in $W$. 

For a string $r$,
$Q(r)$ is the set of suffixes of $r$ that belong to $\p(L_x)$.

\paragraph{Constructing branching program $B_x^c$}
We now present an algorithm which on input $x \in \{0,1\}^n$,runs in time polynomial in $n$ and outputs a branching program $B_x^c$
that recognizes $L_x^c$.

\begin{definition}{\bf Branching program $B_x^c$}

\begin{enumerate}
\item Nodes at level $j$ are triples $(j,s,b)$ where $s \in \p(L_x)$ and $b \in \{0,1\}$. 
(We want string $s$ to be  the longest suffix of $z_1z_2\ldots z_j$
that belongs to $\p(L_x)$, and $b=1$ iff $z_1z_2\ldots z_j$ contains a substring that belongs to $L_x$.)
\item The start node is $(0,\Lambda,0)$ where $\Lambda$ is the
empty string. 
\item  The accepting nodes $(n,s,b)$ are those with $b=1$.
\item For $j \leq n$, the arc out of nodes $(j-1,s,b)$ labeled by alphabet
symbol $\alpha$ is  $(j,s',b')$ where $s'$ is the
longest string in $Q(s\alpha)$ and $b'=1$ if $s'$ contains a suffix
in $L_x$ and otherwise $b'=b$.
\end{enumerate}
\end{definition}
It is clear that the branching program can be constructed (as a directed
graph)
in time polynomial in $n$.   It remains to show that it accepts
those $z$ that have a substring that belongs to $(L_x)$.
 
Fix a string $z \in \{0,1\}^n$. Let $(j,s_j,b_j)$ be the
$j$th vertex visited by the branching program on input $z$.
Note that $s_j$ is  a suffix
of $z_1 \ldots z_j$.  Let $h_j$ be the index
such that $s=z_{h_j} \ldots z_j$; if $s$ is empty, we set $h_j=j+1$.
For $j$ between 1 and $n$ let $i_j$ be the least index
such that $z_{i_j} \ldots z_j$ belongs to $\p(L_x)$ (so
$i_j=j+1$ if there is no such string).    Note that $i_j \geq i_{j-1}$
since if $z_i \ldots z_j$ belongs to $\p(L_x)$ so does $z_i \ldots z_{j-1}$.

The branching program is designed to make the following true:

\begin{claim}
For $j$ between 1 and $n$, $h_j=i_j$ and
$b_j=1$
if and only
if a substring of $z_1 \ldots z_j$ belongs to $L_x$.
\end{claim}

The claim for $b_j=1$ implies that the branching program accepts
the desired set of strings.

\begin{proof}
The claim follows easily by induction, where the basis $j=0$ is
trivial. Assume $j>0$. First we show that $h_j=i_j$. By induction  $h_{j-1}=i_{j-1}$ and by definition of $h_j$ and $i_j$ we have
$i_j \leq h_j$.  To show $h_j \leq i_j$, note that
since $i_j \geq i_{j-1}=h_{j-1}$, the string $z_{i_{j}}\ldots z_j$
is in $Q(t_{j-1}\alpha)$ and so is considered in the choice of $s_j$
and thus $h_j=i_j$.

For the claim on $b_j$, if $z$ has no substring in $L_x$ then
$b_j$ remains 0 by induction.  If $z$ has a substring in $L_x$
let $z_i \ldots z_k$ be such a substring with $k$ minimum.
Then by the claim on $t_k$, $h_k\leq i$, and so
$z_i \ldots z_k$ is a suffix of $s_k$ and so $b_k=1$,
and for all $j \geq k$, $b_j$ continues to be 1.
\end{proof}

\paragraph{Constructing branching program $B_x^w$}
We now present an algorithm which on input $x \in \{0,1\}^n$,runs in time polynomial in $n$ and outputs a branching program $B_x^w$
that accepts the set of strings $z$ that have a nonempty
suffix $u$ and nonemtpy prefix $v$ such that $uv$ belongs to $L_x$.

\begin{definition}{\bf Branching program $B_x^w$}
\begin{enumerate}
\item Nodes at level $j$ are triples $(j,s,p)$ where $p,s \in \p(L_x)$.
(String $s$ will be  the longest suffix of $z_1z_2\ldots z_j$
that belongs to $\p(L_x)$ (as in $B_x^c$) 
and $p$ is the longest prefix of $z_1z_2 \ldots z_j$
that belongs to $\p(L_x)$.
\item The start node is $(0,\Lambda,\Lambda)$ where $\Lambda$ is the
empty string. 
\item  The accepting states are those states $(n,s,p)$ such that $p$
has a nonempty prefix $p'$ and $s$ has a nonempty suffix $s'$ such that
$s'p' \in L_x$.
\item For $j \leq n$, the arc out of state $(j-1,s,p)$ labeled by alphabet
symbol $\alpha$ is  $(j,s',p')$ where $s'$ is the
longest string in $Q(s\alpha)$ and $p'=p\alpha$ if $|s|=j-1$
and $p\alpha \in \p(L_x)$ and $p'=p$ otherwise.
\end{enumerate}

\end{definition}
It is clear that the branching program can be constructed (as a directed
graph)
in time polynomial in $n$.   It remains to show that it accepts
$L_x^w$.

Fix a string $z \in \{0,1\}^n$. Let $(j,s_j,p_j)$ be the
$j$th node visited by the branching program on input $z$.
Notice that $s_j$ is calculated the same way in $B_x^w$ as
in $B_x^c$ and so $s_j$ is the longest suffix of $z_1\ldots z_j$
that belongs to $\p(L_x)$.

An easy induction shows that $p_j$ is the longest prefix
of $z_1\ldots z_j$ belonging to $\p(L_x)$: Let $k$ be the length
of  the longest prefix of $z$ belonging to $\p(L_x)$.
For $j \leq k$ we have $p_j=z_1 \ldots z_j$ and
for $j>k$, $p_j=z_1 \ldots z_k$.

Finally, we need to show that the branching program
accepts $z$ if and only if $z$ has a a nonempty suffix $s'$
and $z$ has a nonempty prefix $p'$ such that $s'p' \in L_x$.
If the program accepts then the acceptance condition
and the fact that $s_n$ is a suffix of $z$ and $p_n$ is a prefix
of $z$ implies that $z$ has the required suffix and prefix. 
Conversely, if $z$ has such a prefix $p'$ and suffix $s'$,
then they each belong to $\p(L_x)$. Since $p_n$
is the longest prefix of $z$ belonging to $\p(L_x)$,
$p'$ is a prefix of $p_n$ and since $s_n$ is the longest suffix of $z$
belonging to $\p(L_x)$, $s'$ is a suffix of $t_n$.
So the branching program will accept.

\paragraph{Putting things together}
From the constructions, it is clear that the size of the branching programs $B_x^w$ and $B_x^c$ are polynomial in the size of $L_x$ and hence polynomial in $n = |x|$.  Moreover, by a product construction, we can efficiently construct the deterministic finite branching program 
$B_x$ which accepts the strings accepted by $B_x^w$ or $B_x^c$,
which is  $G_{x, \leq n}$. This observation, along with  Lemma~\ref{lem:countautomaton} implies the following lemma. 

\begin{lem}~\label{lem:mainlem1}
There is an algorithm which takes as input a string $x$ in $\{0,1\}^n$ and outputs the size of $G_{x, \leq n}$ in time polynomial in $n$.
\end{lem}

\subsubsection{Computing $|G_{x, \leq p}|$ efficiently}
In this section, we will show that for every $p|n$, we can compute the quantity $|G_{x, \leq p}|$ efficiently. The algorithm will be a small variation of our algorithm for computing $|G_{x, \leq n}|$ from the previous section. Let $p$ be a divisor of $n$ with $p < n$. Every string $y \in G_{x, \leq p}$ is of the form $a^{\frac{n}{p}}$ for some $a \in \{0,1\}^p$, and every string in $\Orbit(y)$ is of the form $(\Rot^i(a))^{\frac{n}{p}}$, for some $i \leq p$. Let us write the string $x$ as $x_1x_2\ldots x_{\frac{n}{p}}$ where for each $i$, $x_i$ is of length exactly $p$. We will now try to characterize the strings in $G_{x,\leq p}$.
From the definitions, $y = a^{\frac{n}{p}} \in G_{x, \leq p}$  if and only if there is a rotation $0 \leq i < p$ such that $(\Rot^i(a))^{\frac{n}{p}}$ has a prefix in $L_x$.  This, in turn, can happen if and only if there is an $i < p$ such that one of the following is true.
\begin{itemize}
\item $\Rot^i(a) < x_1$ in~\lex order, or
\item there is $j$, $0 < j < \frac{n}{p}$, such that $\Rot^i(a) = x_1 = x_2 = x_3 = \ldots = x_i$ and $\Rot^i(a) < x_{i+1}$ in~\lex order.
\end{itemize}
The strings $y = a^{\frac{n}{p}}$ for which $a$ has a rotation which is less than $x_1$ in~\lex order are exactly the strings of the form $c^{\frac{n}{p}}$ with $c \in G_{x_1, \leq p}$. Via the algorithm of the previous subsection, there is a polynomial in $n$ time algorithm which outputs a branching program recognizing $G_{x_1, \leq p}$. The only strings which satisfy the second condition are of the form ${c}^{\frac{n}{p}}$, where $c$ is a rotation of $x_1$ and $x_1 < x_{i+1}$ in~\lex order.
There are at most $|\Orbit(x_1)|$ such strings, and we can count them directly given $x$.

This gives us our algorithm for computing $|G_{x, \leq p}|$:\\
\noindent {\bf Computing $|G_{x, \leq p}|$:}\\
\noindent{\bf Input:} 
\begin{itemize}
\item Integers $n, p$ such that $p|n$
\item A string $x \in \{0,1\}^n$
\end{itemize}  
\noindent{\bf Algorithm:}
\begin{enumerate}
\item Write $x$ as $x = x_1x_2\ldots x_{\frac{n}{p}}$ where $|x_i| = p \forall i \in [\frac{n}{p}]$
\item Construct a branching program $A_{x_1}$ such that $L(A_{x_i}) \cap \{0,1\}^p = G_{x_1, \leq p}$
\item Let $M$ be the number of strings of length $p$ accepted by $A_{x_1}$
\item If there is an $0 < i < \frac{n}{p}$ such that $x_1 = x_2 = x_3 = \ldots x_i$ and $x_1 < x_{i+1}$ in~\lex order,
 and $x_1 \notin L(A_{x_1)}$, then output $M+|\Orbit(x_1)|$, else output $M$.
\end{enumerate}

From the construction in Section~\ref{sec:binalphabet} and Lemma~\ref{lem:mainlem1}, it follows that we can construct $A_{x_1}$ and count $M$ in time polynomial 
in $n$. We thus have the following lemma. 

\begin{lem}~\label{lem:divisors}
For any divisor $p$ of $n$ and string $x\in \{0,1\}^x$, we can compute the size of the set $G_{x, \leq p}$ in time $\poly(n)$. 
\end{lem}

We now have all the ingredients for the proof of the following theorem, which is a special case of Theorem~\ref{thm:indexing} when the alphabet under consideration is $\{0,1\}$. 
\begin{theorem}
There is an algorithm for indexing necklaces of length $n$ over the alphabet $\{0,1\}$,
which runs in time $\poly(n)$.
\end{theorem}

\begin{proof}
The proof simply follows by plugging together the conclusions of Lemma~\ref{lem: 2}, Lemma~\ref{lem: 3}, Lemma~\ref{lem:4}, Lemma~\ref{lem:countautomaton} and Lemma~\ref{lem:divisors}.
\end{proof}

It is not difficult to see that the indexing algorithm can be used to obtain a reverse indexing algorithm as well and hence, we also obtain a special case of Theorem~\ref{thm:reverse_indexing} for the binary alphabet. 

\subsubsection{Indexing necklaces over large alphabets}

In this subsection we how to handle the case of general alphabets
$\Sigma$ (with $|\Sigma| = q$). A direct generalization
of the algorithm for the case of the binary alphabet, where the set $L_x$ is appropriately defined, will run in time polynomial in $n$ and $q$. 
Our goal here is to improve the running time to polynomial in $n$ and $\log q$. 

The basic idea is to represent the elements in $\Sigma$ by binary strings of length $t \defeq \lceil \log q \rceil$.
Let $\bin : \Sigma \to \{0,1\}^t$ be an injective map whose image is the set $\Gamma$ of $q$ lexicographically smallest strings in $\{0,1\}^t$.
Extend this to a map $\bin : \Sigma^n \to \{0,1\}^{tn}$ in the natural way.

We now use the map $\bin$ to convert our indexing/counting problems over the large alphabet $\Sigma$ to a
related problem over the small alphabet $\{0,1\}$.
For $x \in \Sigma^n$, we have $\bin(\Rot^i(x)) = \Rot^{ti}(\bin(x))$.
For an orbit $E \subseteq \Sigma^n$ and $x \in \{0,1\}^{tn}$,
we say $E < x$ if some element $z \in E$ satisfies $\bin(z) \lexq x$.

Let $\Classes_x$ be the set of orbits $E \subseteq \Sigma^n$ which
are less than $x$.
For each $x \in \{0,1\}^{tn}$ and $p \mid n$, define:
\begin{enumerate}
\item $$G_{x, p} = \bigcup_{E < x, |E| = p} E.$$
\item $$ G_{x, \leq p} = \bigcup_{E < x, |E| \mbox{ divides } p} E.$$
\end{enumerate}

The following identity allows us to count $G_{x, \leq n}$:
$$|G_{x, \leq n}| = |\{ y \in \{0,1\}^{tn} \mid  y \in \Gamma^n, \exists i < n \mbox{ s.t. } \Rot^{it}(y) \lexq x  \}|.$$
It is easy to efficiently produce a branching program $A_0$ such that
$L(A_0) \cap \{0,1\}^{tn} = \Gamma^n$.
As we will describe below, the methods of the previous section
can be easily adapted to efficiently produce a branching program
$A_x$ such that
$$L(A_x) \cap \{0,1\}^{tn} = \{ y \in \{0,1\}^{tn} \mid  \exists i < n \mbox{ s.t. } \Rot^{it}(y) \lexq x \}.$$

The following lemma will be crucial in the design of this branching program.
\begin{lem}~\label{lem:largesigma}
Let $y \in \{0,1\}^{tn}$. There exists $i < n$ such that $\Rot^{it}(y) \lexq x$ 
if and only if at least one of the following events occurs:
\begin{enumerate}
\item there exists $w \in L_x$ such that $w$ appears as a contiguous substring of $y$ starting at a coordinate $j$ with $j \equiv 0 \mod t$ (where the coordinates of $x$ are $0,1, \ldots, (tn-1)$).
\item there exist strings $w_1, w_2$ such that $w_1w_2 \in L_x$,  
$w_2$ is a prefix of $y$, $w_1$ is a suffix of $y$, and $|w_1| \equiv 0 \mod t$.
\end{enumerate}
\end{lem}


Given this lemma, the construction of $A_x$ follows easily via
the techniques of the previous subsections. The main addition
is that one needs to remember the value of the current
coordinate mod $t$, which can be done by blowing up the number of states
of the branching program by a factor $t$.

Intersecting the accepted sets of $A_x$ and $A_0$ gives us our desired branching program
which allows us to count $|G_{x, \leq n}|$.
This easily adapts to also count $|G_{x, \leq p}|$ for each $p \mid n$.

We conclude using the ideas of Section~\ref{sec:reduc}.
We can now compute $|G_{x, p}|$ for each $x$ and each $p \mid n$. From Lemma~\ref{lem: 2}, Lemma~\ref{lem: 3} and Lemma~\ref{lem:4}, it follows that for every $x$, we can compute $|\Classes_x|$ efficiently.  We thus get our main indexing theorem for necklaces from Lemma~\ref{lem:reduction}.


\begin{theorem}~\label{thm:largesigma}
There are $\poly(n, \log|\Sigma|)$-time indexing and reverse-indexing algorithms for necklaces of length $n$ over $\Sigma$.

Furthermore, there are $\poly(n, \log|\Sigma|)$-time indexing and reverse-indexing algorithms for necklaces of length $n$ over $\Sigma$ with fundamental period exactly $n$.
\end{theorem}

\section{Indexing irreducible polynomials}~\label{sec:irrind}
In the previous section, we saw an algorithm for indexing necklaces of length $n$ over an alphabet $\Sigma$ of size $q$, which runs in time polynomial in $n$ and $\log q$. In this section, we will see how to use this algorithm to efficiently index irreducible polynomials over a finite field.
More precisely, we will use an indexing algorithm for necklaces with fundamental period exactly equal to $n$ (which is also given
by the methods of the previous sections).

Let $q$ be a prime power, and let $\F_q$ denote the finite field of $q$ elements.
For an integer $n > 0$, let $I_{q,n}$ denote the set of monic, irreducible polynomials of degree $n$ in $\F_q[T]$.

\begin{theorem}
For every $q, n$ as above,
there is an algorithm that runs in $\poly(n, \log q)$ time,
takes $O(n \log q)$ bits of advice,
and indexes $I_{q,n}$.
\end{theorem}
\begin{proof}
To prove this theorem, we start by first describing the connection between the tasks of indexing necklaces and indexing irreducible polynomials. 
Let $P(T) \in I_{q,n}$.
Note that $P(T)$ has all its roots in the field $\F_{q^n}$.
Let $\alpha \in \F_{q^n}$ be one of the roots of $P(T)$.
Then we have that $\alpha, \alpha^q, \ldots, \alpha^{q^{n-1}}$
are all distinct, and:
$$ P(T) = \prod_{i=0}^{n-1} (T - \alpha^{q^{i}}).$$

Conversely, if we take $\alpha \in \F_{q^n}$
such that $\alpha, \alpha^q, \ldots, \alpha^{q^{n-1}}$ are all
distinct, then the polynomial 
$P(T)= \prod_{i=0}^{n-1} (T - \alpha^{q^i})$
is in $I_{q,n}$.

Define an action of $\Z_n$ on $\F_{q^n}^*$ as follows:
for $k \in \Z_n$ and $\alpha \in (\F_{q^n})^*$, define:
$$k [\alpha] = \alpha^{q^k}.$$
This action partitions $\F_{q^n}^*$ into orbits.
By the above discussion, $I_{q,n}$ is in one-to-one correspondence
with the orbits of this action with size exactly $n$.
Thus it suffices to index these orbits.

Let $g$ be a generator of the the multiplicative group $(\F_{q^n})^*$.
Define a map $E: \Z_{q^n-1} \to \F_{q^n}^*$ by:
$$ E(a) = g^a.$$
We have that $E$ is a bijection.
Via this bijection, we have an action of $\Z_n$ on $\Z_{q^n-1}$, 
where for $k \in \Z_n$ and $a \in \Z_{q^n - 1}$, 
$$ k[a] = q^k \cdot a.$$

Now represent elements of $\Z_{q^n -1}$ by integers
in $\{0, 1, \ldots, q^n - 2\}$. Define $\Sigma = \{0,1, \ldots, q-1\}$.
For $a \in \Z_{q^n-1}$, consider its base-$q$ expansion
$a_\sigma \in \Sigma^n$. This gives us a bijection between
$\Z_{q^n - 1}$ and $\Sigma^n \setminus \{(q-1, \ldots, q-1) \}$.
Via this bijection, we get an  action of $\Z_n$
on $\Sigma^n \setminus \{(q-1, \ldots, q-1) \}$.
This action is precisely the standard rotation action!

This motivates the following algorithm.\\
\noindent {\bf \underline{The Indexing Algorithm:}}\\
\noindent{\bf Input:} $q$ (a prime power), $n \geq 0$, $i \in [ |I_{q,n}| ]$\\
\noindent{\bf Advice:} 1. A description of $\F_q$\\
\noindent 2. An irreducible polynomial $F(T) \in \F_q[T]$
of degree $n$, whose root is a generator $g$ of $(\F_{q^n})^*$ (a.k.a. primitive polynomial).
\begin{enumerate}
\item Let $\Sigma = \{0,1, \ldots, q-1\}$.
\item Use $i$ to index an necklace $\sigma \in \Sigma^n \setminus \{ (q-1, q-1, \ldots, q-1) \}$ with fundamental
period exactly $n$ (via Theorem~\ref{thm:largesigma}).
\item View $\sigma$ as the base $q$ expansion of an integer $a \in \{0,1, \ldots, q^n - 2\}$.
\item Use $F(T)$ to construct the finite field $\F_{q^n}$ and the element $g \in \F_{q^n}^*$. (This can be done by setting $\F_{q^n} = \F_q[T]/F(T)$, and taking the class of the element $T$ in that quotient to be the element $g$.)
\item Set $\alpha = g^a$.
\item Set $P(T) = \prod_{i= 0}^{n-1} (T -\alpha^{q^i})$.
\item Output $P(T)$.
\end{enumerate}

For constant $q$, this algorithm can be made to work with $\poly(\log n)$ advice.
Indeed, one can construct the finite field $\F_{q^n}$ in $\poly(q, n)$ time,
and a wonderful result of Shoup~\cite{Shoup} constructs a set of 
$q^{\poly(\log n)}$ elements in $\F_{q^n}$, one of which is guaranteed to be a generator.
The advice is then the index of an element of this set which is a generator.
\end{proof}

\section{Explicit Generator Matrices and Parity Check Matrices for BCH codes}
\label{sec:BCH}
In this section, we will apply the indexing algorithm for necklaces
to give a strongly explicit construction for generator and the parity check matrices for BCH codes. 
More precisely, we use the fact that our indexing algorithm is in fact an 
unranking algorithm for the lexicgraphic ordering on (lexicographically least representatives of) necklaces.

BCH codes~\cite{MS78} are classical algebraic error-correcting codes based on polynomials over finite extension fields. They have played a central role since the early days of coding theory due to their remarkable properties (they are one of the few known families of codes that has better rate/distance tradeoff than random codes in some regimes). Furthermore, their study  motivated many advances in algebraic algorithms.

Using our indexing algorithm for necklaces, we can answer a basic
question about BCH codes: we construct strongly explicit explicit generator matrices and parity check matrices for BCH codes. For the traditionally used setting of parameters
(constant designed distance), it is trivial to construct generator matrices and parity check matrices for BCH codes. But for large values of the designed
distance, as far as we are aware, this problem was unsolved.

Let $q$ be a prime power, and let $n\geq 1$ and $0 \leq d < q^n-1$.
The BCH code associated with these parameters will be of length $q^n$ over the
field $\F_q$, where the $q^n$ coordinates are identified with the big field $\F_{q^n}$.
Let:
$$ V = \{ \langle P(\alpha) \rangle_{\alpha \in \F_{q^n}} \mid P(X) \in \F_{q^n}[X], \deg(P) \leq d, \mbox{ s.t. } \forall \alpha \in \F_{q^n}, P(\alpha) \in \F_q \}.$$
In words: this is the $\F_q$-linear space of all $\F_{q^n}$-evaluations of $\F_{q^n}$-polynomials of low degree, which have the property that all their evaluations lie in $\F_q$. In coding theory terminology, this is a subfield subcode of Reed-Solomon codes.

The condition that $P(\alpha) \in \F_q$ for each $\alpha \in \F_{q^n}$
can be expressed as follows:
$$ P(X)^q = P(X)  \mod X^{q^n} - X.$$
Thus, if $P(X) = \sum_{i=0}^d a_i X^i$,
then the above condition is equivalent to:
$$ \sum_{i=0}^d a_i^q X^{iq} = \sum_{i=0}^d a_i X^i \mod X^{q^n} - X,$$
which simplifies to:
$$ \forall i, a_{iq \mod (q^n - 1)} = a_i^q.$$
Thus:
\begin{enumerate}
\item For every $i$, if $\ell$ is the smallest integer such that $iq^\ell \mod (q^n-1) = i$, then $a_i \in V_\ell = \{ \alpha \in \F_{q^n} \mid \alpha^{q^\ell} = \alpha \}$,
\item Specifying $a_i \in V_\ell$ automatically determines $a_{iq  \mod (q^n - 1) }, a_{iq^2 \mod (q^n-1) }, \ldots $,
\item $a_i$ can take any value in $V_\ell$.
\end{enumerate}

This motivates the following choice of basis for BCH codes.
Let $\mathcal F = \{ S \subseteq \{0,1,\ldots, d\} \mid  i \in S \Rightarrow (iq \mod (q^n-1)) \in S \}.$
Let $\alpha_{S,1}, \ldots, \alpha_{S, |S|}$ be a basis for $V_{|S|}$ over $\F_q$
(note that when $j \mid n$,
we have that $V_\ell = \{ \alpha \in \F_{q^n} \mid \alpha^{q^\ell} = \alpha \}$
is an $\F_q$-linear subspace of $\F_{q^n}$ of dimension $\ell$).
For $S \in \mathcal F$, define $m_S = \min_{i \in S} i$.
For $S \in \mathcal F$ and $j \in [|S|]$,
define:
$$ P_{S, j} (X) = \sum_{k = 0}^{|S|-1} \alpha_j^{q^k} X^{m_S q^{k} \mod (q^n-1)}.$$
It is easy to see from the above description that $\left( P_{S, j} \right)_{S \in \mathcal F, j \in [n]}$
forms an $\F_q$ basis for the BCH code $V$. Thus it remains to show that one can index the sets of $\mathcal F$.

If we write all the elements of $S \in \mathcal F$ in base $q$, 
we soon realize that the $S$ are precisely in one-to-one correspondence with those
rotation orbits of $\Sigma^n$ (with $\Sigma = \{0,1, \ldots, q-1\}$) where all elements
of the orbit are lexicographically $\leq$ some fixed string in $\Sigma^n$ (in this case the fixed string
turns out to be the base $q$ representation of the integer $d$). By our indexing algorithm for orbits,
$\mathcal F$ can be indexed efficiently. Thus we can compute any given entry of a generator matrix for BCH codes.

The parity check matrices can be constructed similarly.
For a given designed distance $d$, one starts with 
$d \times \F_{q^n}^*$ matrix $M$ whose $i, \alpha$
entry equals $\alpha^i$. Note that every $d$ columns of this matrix
form a van der Monde matrix: thus they are linearly independent over
$\F_{q^n}$ (and hence also over $\F_q$). 

Define an equivalence $\sim$ relation on $[d]$ as follows:
$i_1 \sim i_2$ iff $i_2 = i_1 \cdot q^k \mod (q^n-1)$ for some $k$.
Now amongst the rows of $M$, for each equivalance class $E \subseteq d$,
keep only one row from $E$ (i.e., for some $i \in E$, keep the $i$'th
row of $M$ and delete the $j$'th row for all $j \in E \setminus \{i\}$).
The remarkable dimension-distance tradeoff of BCH codes is based
on the fact that this operation, while it reduces the dimension 
of the ambient space in which the columns of this matrix lie,
preserves the property that every $d$ columns of this matrix are linearly
independent over the small field $\F_q$.
This reduced matrix $\tilde{M}$ is the parity-check matrix of the BCH 
code.

We now give a direct construction of the parity-check
matrix $\tilde{M}$. Let $\mathcal F = \{ S \subseteq [q^n-1] \mid
i \in S \implies iq \in S \}$. For $S \in \mathcal F$,
let $m_S = \min_{i \in S} i$. Then the rows of $\tilde{M}$
are indexed by those $S \in \mathcal F$ for which
$m_S \leq d$. The $(S, \alpha)$ entry of $\tilde{M}$
equals $\alpha^{m_S}$. Writing all the integers of $[q^n-1]$
in base $q$, we see that the elements of $\mathcal F$ are
orbits of the $\Z_n$ action on $\Sigma^n$, where $\Sigma
= \{0, 1, \ldots, q-1\}$. Furthermore, the $S$ with $m_S \leq d$
are precisely those orbits which have some element lexicographically
at most a given fixed element $x$ (which in this case is the 
base $q$ representation of $d$). By our indexing algorithm,
the rows of $\tilde{M}$ can be indexed efficiently, and hence
each entry of the $\tilde{M}$ can be computed in time
$\poly(n)$, as desired.


\section{Open Problems}\label{sec:openprobs}

We conclude with some open problems.

\begin{enumerate}

\item Can the orbits of group actions be indexed in general?

One formulation of this problem is as follows: Let $G$ be a finite group acting on a set $X$, both of size $\poly(n)$.
Suppose $G$ and its action on $X$ are given as input explicitly.
For a finite alphabet $\Sigma$, consider the action of $G$ on $\Sigma^X$ (by
permuting coordinates according to the action on $X$). Can the orbits of this action
be indexed? Can they be reverse-indexed?

\item Let $G$ be the symmetric group $S_n$. Consider its action on $\{0,1\}^{{[n] \choose 2}}$,
where $G$ acts by permuting coordinates. The orbits of this action correspond to the isomorphism classes of $n$-vertex graphs. Can these orbits be indexed?

More ambitiously, can these orbits be reverse-indexed? This would imply that graph isomorphism is in $P$.

\item It would be interesting to explore the complexity theory
of indexing and reverse-indexing. Which languages can be indexed efficiently?
Can this be characterized in terms of known complexity classes?

In particular, it would be nice to disprove the conjecture:
``Every pair-language $L \in P$ for which the counting problem can be solved efficiently can be efficiently indexed".

\end{enumerate}

\section*{Acknowledgements}
We would like to thank Joe Sawada for making us aware of the work of Kociumaka et al~\cite{KRR14}.

\bibliographystyle{alpha}
\bibliography{refs}
\appendix
\section{Alternative indexing algorithm for binary necklaces of prime length}
\label{sec:magic}

In this section we give another algorithm for indexing necklaces in $\{0,1\}^n$ 
in the special case where $n$ is prime.

For convenience, we will denote the $n$ coordinates of $\{0,1\}^n$
by $0,1, \ldots, n-1$, and identify them with elements of $\Z_n$.
\begin{definition}
Let $x \in \{0,1\}^n$. We say $x$ is top-heavy  if for every
$j$, $0 \leq j < n$:
$$ \sum_{k = 0}^j \left( x_k - \frac{wt(x)}{n} \right) \geq 0.$$
\end{definition}
In words: every prefix of $x$ has normalized Hamming weight at
least as large as the normalized Hamming weight of $x$.

The next lemma by Dvoretzky and Motzkin~\cite{DM47}  shows that every string has a unique top-heavy rotation.
\begin{lem}[\cite{DM47}]
Let $n$ be prime.
For each $x \in \{0,1\}^n \setminus \{0^n, 1^n\}$, there exists a unique $i$, $0\leq i < n$ such that 
$\Rot^i(x)$ is top-heavy.
\end{lem}
\begin{proof}
Define $f: \{0,1 \}^n  \times \mathbb{N} \to \mathbb R$ by:
$$f(x, j) = \sum_{k = 0}^j \left(x_{k \mod n} - \frac{\wt(x)}{n} \right).$$
Then the top-heaviness of $x$ is equivalent to $f(x, j) \geq 0$ for all
$j \in \mathbb N$.

We make two observations:
\begin{enumerate}
\item If $j = j' \mod n$, then  $f(x, j) = f(x, j')$.
This follows from the fact that:
$$\sum_{k = 0}^{n-1} \left(x_{k} - \frac{\wt(x)}{n} \right) = 0.$$
\item For nonnegative integers $j, \ell$ with $j < n$, we have:
$$f(\Rot^j(x), \ell) = f(x, j + \ell) - f(x, j).$$
\end{enumerate}

Putting these two facts together, we get that:
\begin{align}
\label{eqtop}
f(\Rot^j(x), \ell) = f(x, (j+\ell) \mod n) - f(x, j).
\end{align}

Now fix $x \in \{0,1\}^n \setminus \{0^n, 1^n\}$.
Define $i \in \{0,1, \ldots, n-1\}$ to be such that
$f(x, i)$ is minimized.
By Equation~\eqref{eqtop}, we get that
$f(\Rot^i(x), \ell) \geq 0$ for all nonnegative integers $\ell$.
This proves the existence of $i$.

For uniqueness of $i$, we make two more observations:
\begin{enumerate}
\item If $f(x, j) > f(x, i)$, then
$$f(\Rot^j(x), n + i-j) = f(x, n+i) - f(x,j) = f(x, i) - f(x,j) < 0,$$
and thus $\Rot^j(x)$ is not top-heavy.
 
\item If $f(x, j) = f(x,j')$, then $j = j' \mod n$.
To see this, first note that we may assume $j < j'$.
Then:
\begin{align*}
0 &= f(x, j') - f(x,j)\\
&= \sum_{k = j+1}^{j'} \left(x_{k \mod n} - \frac{\wt(x)}{n} \right)\\
&=  \left(\sum_{k = j+1}^{j'} x_{k \mod n} \right)  - (j'-j) \cdot\frac{\wt(x)}{n}.
\end{align*}
Thus, since the first term is an integer, we must have that $(j' - j) \cdot \wt(x)$ must
be divisible by $n$, and by our hypothesis on $x$, we have that $j' = j \mod n$.
\end{enumerate}
Thus $i \in \{0,1, \ldots, n-1\}$, for which $\Rot^i(x)$ is top-heavy, is unique.
\end{proof}

The above lemma implies that each orbit $E$ contains a unique top-heavy string.
We define the canonical element of $E$ to be that element.

We now show that there is a branching program $A$ such that
$L(A) \cap \{0, 1\}^n$  precisely equals the set of top-heavy strings.
By the discussion in the introduction, this immediately gives an indexing algorithm for orbits of $E$.

How does a branching program verify top-heaviness? 
In parallel, for each $\ell \in \{ 1, \ldots, n-1\}$, the branching program
checks if condition $C_\ell$ holds, where $C_\ell$ is:
$$`` \forall 0 \leq j < n,  \sum_{k = 0}^j x_k \geq \frac{k \cdot \ell}{n} " .$$
At the same time, it also computes the weight of $x$.
At the final state, it checks if $C_{\wt(x)}$ is true.
$x$ is top-heavy if and only if it is true.

This completes the description of the indexing algorithm.

We also know an extension of this approach that can handle $n$ which have $O(1)$ prime factors.
The key additional ingredient of this extension is a new encoding of strings
that enables verification of properties like top-heaviness by automata.

\section{Complexity of  indexing}
\label{sec:complexity}

In this section, we explore some basic questions about the complexity theory of indexing and reverse indexing. We would like to understand what sets can be indexed/reverse-indexed efficiently.

The outline of this section is as follows.
We first deal with indexing and reverse-indexing in a nonuniform setting.
Based on some simple observations about 
what cannot be indexed/reverse-indexed, we make some naive, optimistic conjectures characterizing what is efficiently indexable/reverse-indexable, and then proceed to disprove these conjectures. We then make some natural definitions for indexing and reverse-indexing in a uniform setting, and conclude with some analogous naive, optimistic conjectures. 

\subsection{Indexing and reverse-indexing in the nonuniform setting}

By simple counting, most sets $S \subseteq \{0,1\}^n$ cannot be indexed or reverse-indexed by circuits of size $\poly(n)$.
We now make two naive and optimistic conjectures:
\begin{itemize}
\item If $S \subseteq \{0,1\}^n$ has a $\poly(n)$-size circuit recognizing it, then
there is a $\poly(n)$-size circuit for indexing $S$.
\item If $S \subseteq \{0,1\}^n$ has a $\poly(n)$-size circuit recognizing it, then
there is a $\poly(n)$-size circuit for reverse-indexing $S$.
\end{itemize}
Note that the simple observations about indexing made in the introduction are consistent with these conjectures.

We now show that these conjectures are false (unless the polynomial hierarchy collapses).
Assuming the conjectures, we will give $\Sigma_4$ algorithms
to count the number of satisfying assignments of a given boolean formula $\phi$.
By Toda's theorem~\cite{Toda}, this would imply that the polynomial hierarchy collapses.

Let $S \subseteq \{0,1\}^n$ be the set of satisfying assignments of a given boolean
formula $\phi$ of size $m$ ($m \geq n$). We know that $S$ can be recognized by a circuit of size $m$ (namely $\phi$).
By the conjectures, there are circuits $C_i$ and $C_r$ of size $\poly(m)$ for indexing $S$ and reverse-indexing 
$S$. We will now see that a $\Sigma_4$ algorithm can get its hands on these circuits, and then
use these circuits to count the number of elements in $S$.

\paragraph{Indexing} Consider the $\Sigma_4$ algorithm that does the following on input $\phi$.
Guess a circuit $C:\{0,1\}^n \to \{0,1\}^n \cup \{``\mbox{{\bf too large}}"\}$ of size $\poly(m)$, and an integer $K < 2^n$
and then verify the following properties:
\begin{itemize}
\item for all $i \in [K]$, $C(i) \neq$ {\bf too large} and $\phi(C(i)) = 1$.
\item for all $i \notin [K]$, $C(i) =$ {\bf too large}.
\item for all $x \in \{0,1\}^n$, if $\phi(x) = 1$, then there exists a unique $i \in [K]$
for which $C(i) = x$.
\end{itemize}
If $C = C_i$, and $K = |S|$, then these properties hold.
It is also easy to see that if all these properties hold, then $C$ is an indexing circuit for
$S$, and $K = |S|$. Thus the above gives a $\Sigma_4$ algorithm to compute $|S|$.

\paragraph{Reverse-indexing} Consider the $\Sigma_4$ algorithm that does the following on input $\phi$.
Guess a circuit $C:\{0,1\}^n \to \{0,1\}^n \cup \{``\mbox{{\bf false}}"\}$ of size $\poly(m)$, and an integer $K < 2^n$
and then verify the following properties:
\begin{itemize}
\item for all $x \in \{0,1\}^n$, either ($\phi(x) = 1$ and $C(x) \in [K]$) or ($\phi(x) = 0$ and $C(x)= $ {\bf false}).
\item for all $i \in [K]$, there exists a unique $x \in \{0,1\}^n$ such that $C(x) = i$.
\end{itemize}
If $C = C_r$, and $K = |S|$, then these properties hold.
It is also easy to see that if all these properties hold, then $C$ is a reverse-indexing circuit for
$S$, and $K = |S|$. Thus the above gives a $\Sigma_4$ algorithm to compute $|S|$.

\subsection{Indexing and reverse-indexing in the uniform setting}

We now introduce a natural framework for talking about indexing in the uniform setting.

Let $L \subseteq \Sigma^* \times \Sigma^*$ be a pair-language. For $x \in \Sigma^*$, define
$L_x = \{ y \mid (x,y) \in L \}$. An algorithm $M(x,i)$ is said to be an indexing algorithm for
$L$ if for every $x \in \Sigma^*$, the function $M(x, \cdot)$ is an indexing of the set $L_x$.
An algorithm $M(x,y)$ is said to be a reverse indexing algorithm for $L$ if for every $x \in \Sigma^*$,
the function $M(x, \cdot)$ is a reverse indexing of the set $L_x$. Indexing/reverse-indexing algorithms
are said to be efficient if they run in time $\poly(|x|)$.

We now make some preliminary observations about the limitations of efficient indexing/reverse-indexing.
\begin{enumerate}
\item If $L$ can be efficiently indexed, then the counting problem for $L$ can be solved efficiently
(recall that the counting problem for $L$ is the problem of determining $|L_x|$ when given $x$ as input.
The counting problem can be solved via binary search using an indexing algorithm).
\item If $L$ can be efficiently reverse indexed, then $L$ must be in $P$. Indeed,
the reverse indexing algorithm $M(x,y)$ immediately tells us whether $(x,y) \in L$.
\end{enumerate}

In the absence of any other easy observations, we gleefully made the following optimistic conjectures.
\begin{enumerate}
\item Every pair-language $L \in P$ for which the counting problem can be solved efficiently
can be efficiently indexed.
\item Every pair-language $L \in P$ can be efficiently reverse indexed.
\end{enumerate}
Using ideas similar to those used in the nonuniform case, one can show that the latter of these conjectures is
not true (unless the polynomial hierarchy collapses). However we have been unable to say anything interesting about the first conjecture,
and we leave the conjecture that it is false as an open problem.

\end{document}